\documentclass[preprint,12pt]{elsarticle}

\usepackage{amssymb}
\usepackage{amsmath}
\usepackage{amsthm}
\usepackage{array}
\usepackage{arydshln}
\usepackage{tikz}
\usepackage{algpseudocode}
\usepackage[ruled]{algorithm}
\newtheorem{theorem}{Theorem}
\newtheorem{lemma}{Lemma}
\newdefinition{remark}{Remark}
\newtheorem{example}{Example}
\usepackage{diagbox}
\newdefinition{definition}{Definition}

\newtheorem{observation}{Observation}
\newtheorem{problem}{Problem}
\newtheorem{corollary}{Corollary}
\newcommand{\rsp}{\textsc{Robot Scheduling}~}

\begin{document}

\begin{frontmatter}



\title{Collision-Free Robot Scheduling\footnote{A Preliminary Version of this paper appears at ALGOWIN 2024}}


\author[ST]{Duncan Adamson}
\author[LRC,CS]{Nathan Flaherty}
\author[CS]{Igor Potapov}
\author[CS]{Paul G. Spirakis}

\affiliation[ST]{organization={Department of Computer Science, University of St Andrews}, addressline = {North Haugh}, city={KY18 9SX}, country = {United Kingdom} }

\affiliation[LRC]{organization={Leverhulme Research Centre,University of Liverpool},
            addressline={51 Oxford St}, 
            city={Liverpool},
            postcode={L7 3NY}, 
            country={United Kingdom}}

\affiliation[CS]{organization={Department of Computer Science,University of Liverpool},
            addressline={Ashton Street}, 
            city={Liverpool},
            postcode={L69 3BX}, 
            country={United Kingdom}}

\begin{abstract}
    In this paper, we investigate the problem of designing \emph{schedules} for completing a set of tasks at fixed locations with multiple robots in a laboratory. We represent the laboratory as a graph with tasks placed on fixed vertices and robots represented as agents, with the constraint that no two robots may occupy the same vertex at any given timestep. Each schedule is partitioned into a set of timesteps, corresponding to a walk through the graph (allowing for a robot to wait at a vertex to complete a task), with each timestep taking time equal to the time for a robot to move from one vertex to another and each task taking some given number of timesteps during the completion of which a robot must stay at the vertex containing the task. The goal is to determine a set of schedules, with one schedule for each robot, minimising the number of timesteps taken by the schedule taking the greatest number of timesteps within the set of schedules.  We show that this problem is NP-complete for both star graphs (for $k \geq 2$ robots), and planar graphs (for any number of robots). Finally, we provide positive results for path, cycle, and tadpole graphs, showing that we can find an optimal set of schedules for $k$ robots completing $m$ tasks of equal duration of a path of length $n$ in $O(kmn)$, $O(k m n^2)$ time, and $O(k^3 m^4 n)$ time respectively.
\end{abstract}

%
%
%
%
%
%
\end{frontmatter}

\section{Introduction}

Across a wide range of industries, there is an increase in the use of automation. This has led to a wide range of problems relating to the scheduling of autonomous agents within workplaces. This includes spacecraft manufacturing \cite{Liu2023}, Unmanned Aerial Vehicle \cite{qamar2023trmaxalloc}, and vehicle routing \cite{zhang2024automated}.\looseness=-1

In this paper, we are particularly interested in the scheduling of robots within chemistry labs. This is motivated by a significant and expanding body of work concerning robotic chemists. Initial work on these systems focused on building robots performing reactions within fixed environments \cite{granda2018controlling,king2011rise,langner2019ternary,li2015synthesis,macleod2020selfdriving}, however recently Burger et al. \cite{burger2020mobile} have presented a robot capable of moving within a laboratory and completing tasks throughout the space. The works of Burger et al. \cite{burger2020mobile} and Liu et al. \cite{Liu2023} provide the main motivation for this work, namely the problem of moving robots within a laboratory environment (as presented by Burger et al. \cite{burger2020mobile}) while avoiding collisions (as investigated in the manufacturing context by Liu et al. \cite{Liu2023}).\looseness=-1

In addition to physical science motivation, our model and algorithmic results are strongly based on graph theory, in particular, graph exploration. Informally, we model our problem as a graph problem, where robots are represented as agents in the graph, with the goal of finding a set of walks for each robot, allowing every task to be completed without any collisions. Our model of movement for robots within the graph matches the exploration model given by Czyzowicz et al. \cite{CZYZOWICZ201770}, where agents (robots) start at fixed points within the graph, then can move provided that no pair of agents occupy the same vertex in the same timesteps. The primary difference between our model and that of \cite{CZYZOWICZ201770} is that in our setting, the agents are given a schedule from some central system rather than each having to determine the best route separately.\looseness=-1

More general exploration problems have been considered in a variety of settings. Of particular interest to us are the works regarding the efficient exploration of temporal graphs. As in our setting, exploration is, in most cases, centrally controlled, with the primary goal of minimising the number of timesteps required to complete the exploration, corresponding to the length of the longest walk taken by any agent in the graph. Further, having the edge set of the graph change over time is similar to, and indeed can be closely mimicked by, the collision-avoiding condition in our problem, in the sense that the available moves for a given agent change throughout the lifetime of the graph.

There is a large number of results across many settings and variations of the temporal graph exploration problem, including when the number of vertices an agent can visit in one timestep is unbounded \cite{arrighi2023kernelizing,erlebach2022parameterized}, bounded \cite{erlebach_et_al:LIPIcs.ICALP.2019.141,erlebach2021temporal,michail2016traveling}, and for specific graph classes \cite{adamson2022faster,akrida2021temporal,bodlaender2019exploring,bumpus2023edge,deligkas2022optimizing,erlebach2022exploration,erlebach2018faster,taghian2020exploring}.
Particularly relevant to us is the work of Michail and Spirakis \cite{michail2016traveling}, who showed that the problem of determining the fastest exploration of a temporal graph is NP-hard, and, furthermore, no constant factor approximation algorithm exists of the shortest exploration (in terms of the length of the path found by the algorithm, compared to the shortest path exploring the graph) unless $P = NP$. As noted, the change in the structure of temporal graphs is close to the challenges implemented in our graph by agents blocking potential moves from each other.
In terms of positive results, the work of Erlebach et al. \cite{erlebach2021temporal} provided a substantial set of results that have formed the basis for much of the subsequent work on algorithmic results for temporal graph exploration. Of particular interest to us are the results that show that, for temporal graphs that are connected in every timestep, an agent can visit any subset of $m$ vertices in at most $O(n m)$ time, and provide constructions for faster explorations of graphs with $b$ agents and an $(r, b)$-division ($O(n^2 b / r + n r b^2)$ time), and $2 \times n$ grids with $4 \log n$ agents ($O(n \log n)$ time).
\looseness=-1R
\paragraph{\textbf{Our Contributions}}

In this paper, we present a set of results for the $k$-\rsp problem. A short summary is provided in Table \ref{tab:results}. Informally, we define the $k$-\rsp scheduling problem as the problem of assigning schedules (walks on the graph with robots completing every task from a given set), minimising the time needed to complete the schedule.\looseness=-1

We lay out the remainder of this paper as follows.
In Section \ref{sec:preliminaries} we provide the definitions and notation used in the rest of the paper, with the $k$-\rsp problem fully presented in Problem \ref{prob:robot-scheduling}.
In Section \ref{sec:np-hardness} we show that $k$-\rsp is NP-complete for a large number of graph classes, explicitly Complete Graphs, Bipartite Graphs, Star graphs (Theorem \ref{thm:star-hardness}), and Planar graphs (Theorem \ref{thm:planar-hardness}). Finally, Section \ref{sec:algorithmic-results} provides the algorithmic results for this paper, namely an optimal algorithm for constructing a schedule for $k$ robots on a path, cycle and tadpole graphs for tasks with equal duration (Theorems \ref{thm:k-robots-partition-equal-length-task}, \ref{thm:k-robot-cycle} and \ref{thm:k-rsp-on-tadpoles} respectively), and a $k$-approximation algorithm for creating a schedule with $k$ robots on a path graph (Theorem \ref{thm:k-robot-k-approximation}).\looseness=-1

\begin{table}[]
    \centering
    \begin{tabular}{|c|c|c|}
         \hline
         Setting & Result  \\
         \hline
         General graphs, $k \in \mathbb{N}$ & NP-complete\\ &(Theorem \ref{thm:general-hardness})\\
          \hdashline
         Star graphs (and trees), $k \geq 2 $ & NP-complete\\ & (Theorem \ref{thm:star-hardness})\\
         \hdashline
         Planar graphs, $k \in \mathbb{N}$ & NP-complete\\ & (Theorem \ref{thm:planar-hardness})\\
         \hdashline
         Path graphs, with $m$ tasks of equal duration, & Optimal $O(k m n)$ time algorithm\\ &(Theorem \ref{thm:k-robots-partition-equal-length-task})\\
         \hdashline
         Cycle graphs, with $m$ tasks of equal duration & Optimal $O(k m n^2)$ time algorithm \\ &(Theorem \ref{thm:k-robot-cycle})\\
         \hdashline
         Tadpole graph, with $m$ tasks of equal duration & Optimal $O(k^3 m^4, n)$ time algorithm\\ &(Theorem \ref{thm:k-rsp-on-tadpoles})\\
         \hdashline
         Path graphs, $k \in \mathbb{N} $ & k-approximation Algorithm\\ & (Theorem \ref{thm:k-robot-k-approximation})\\
         \hline
    \end{tabular}
    \caption{Our results for $k$-\rsp for different graph classes and values of $k \in \mathbb{N}$.\looseness=-1}
    \label{tab:results}
\end{table}

\section{Preliminaries}
\label{sec:preliminaries}

For the remainder of this paper, we define graphs as a tuple containing a set of vertices $V$ and a set of edges $E \subseteq V \times V$. A \emph{walk} in a graph $G$ of length $\ell$ is a sequence of $\ell$ edges such that the second vertex in the $i^{th}$ edge is the first vertex in the $(i + 1)^{th}$ edge, i.e. a sequence of the form $(v_1, v_2), (v_2, v_3), \dots, (v_{\ell - 1}, v_{\ell})$. Any walk $w$ can visit the same vertex multiple times and may use the same edge multiple times. Given a walk $w = (v_1, v_2), (v_2, v_3), \dots, (v_{\ell - 1}, v_{\ell})$, we denote by $\vert w \vert$ the total number of edges in $w$, and by $w[i]$ the $i^{th}$ edge in $w$. In this paper, we also allow walks to contain self-adjacent moves, i.e. moves of the form $(v_i, v_i)$ for every vertex in the graph. We do so to represent remaining at a fixed position for some length of time.
Given a pair of naturals $i, j \in \mathbb{N}$ where $i \leq j$, we denote by $[i, j]$ the set $\{i, i + 1, \dots, j\}$. For a given walk $w$, we denote by $w[i, j]$ the walk $w[i], w[i + 1], \dots, w[j]$.\looseness=-1

In this problem, we consider a set of agents, which we call \emph{robots}, moving on a given graph $G = (V, E)$ and completing a set of tasks $\mathcal{T} = \{t_1, t_2, \dots, t_m\}$. As mentioned in our introduction, this problem originates in the setting of lab spaces, particularly in the chemistry setting. As such, our definitions of robots and tasks are designed to mimic those found in real-world problems. We associate each task with a vertex on which it is located and the duration required to complete the task. We do not allow tasks to be moved by a robot, a task can only be completed by a single robot remaining at the station for the entire task duration, and any robot may complete any number of tasks, with no restrictions on which task a robot can complete. This requirement reflects the motivation from chemistry, where tasks reflect reactions that must be done within an exact time frame and at a fixed workstation.\looseness=-1

Formally, we define a task $t_i$ as a tuple $(v_i,d_i)$ where $d_i$ is the \emph{duration} of the task, and $v_i$ is the vertex at which the task is located. We use $\vert t_i \vert$ to denote the duration of the task $t_i$. In general, the reader may assume that for a graph $G = (V, E)$ containing the vertex set $V = \{v_1, v_2, \dots, v_n\}$, the notation $i_{t}$ is used to denote the index of the vertex at which task $t = (v_{i_t},d)$ is located. This will be specified throughout the paper where relevant.\looseness=-1

To complete tasks, we assign each robot a \emph{schedule}, composed of an alternating sequence of walks and tasks. We note that each schedule can begin and end with either a walk and a walk, a walk and a task, a task and a walk, or a task and a task. We treat each schedule as a set of commands to the robot, directing it within a given time frame. In this way, we partition the schedule into a set of timesteps, with each timestep allowing a robot to move along one edge or complete some fraction of a task, with a task $t$ requiring exactly $\vert t \vert$ timesteps to complete. We call the time span of a schedule the total number of timesteps required to complete it. The \emph{time span} of the schedule $C$ containing the walks $w_1, w_2, \dots, w_{p}$ and tasks $t_1, t_2, \dots, t_m$  is given by $\vert C \vert = \left( \sum_{i \in [1, p]} \vert w_i \vert \right) + \left(\sum_{j \in [1, m]} \vert t_j \vert \right)$. For a set of schedules $\mathcal{C}$ the time span is given by the maximum time span of all schedules in $\mathcal{C}$. Given a walk $w$ directly following the task $t$ in the schedule $C$, we require that the first edge traversed in $w$ begins at the vertex $v_{i_t}$ on which $t$ is located. Similarly, we require that the task $t'$ following the walk $w'$ in the schedule $C$ is located on the last vertex in the last edge in $w'$. We additionally assume that the robot remains on the last vertex visited in the schedule.\looseness=-1

The \emph{walk representation} $\mathcal{W}(C)$ of a schedule $C$ is an ordered sequence of edges formed by replacing the task $t_i = (v_i,d)$ in $C$ with a walk of length $\vert t_i \vert = d$ consisting only of the edge $(v_{i}, v_i)$, then concatenate the walks together in order. Note that $\vert \mathcal{W}(C) \vert = \vert C \vert$. For a given robot $R$ assigned schedule $C$, in timestep $i$ $R$ is located on the vertex $v \in V$ that is the end vertex of the $i^{th}$ edge in $\mathcal{W}(C)$, i.e., the vertex $v$ such that $\mathcal{W}(C)[i] = (u, v)$. We require the first vertex in the walk representation of any schedule $C$ assigned to robot $R$ to be the \emph{starting vertex} of $R$, i.e. some predetermined vertex representing where $R$ starts on the graph. 
If the schedule $C$ containing the task $t$ is assigned to robot $R$, we say that $t$ is \emph{assigned} to $R$.\looseness=-1

Given a set of schedules $\mathcal{C} = (C_1, C_2, \dots, C_{k})$ for a set of $k$ robots $R_1, R_2, \dots$, $R_k$, and set of tasks $\mathcal{T} = (t_1, t_2, \dots, t_m)$. we say that $\mathcal{C}$ is \emph{task completing} if for every task $t \in \mathcal{T}$ there exists exactly one schedule $C_i$ such that $t \in C_i$. We call $\mathcal{C}$ \emph{collision-free} if there is no timestep where any pair of robots occupy the same vertex or traverse the same edge. Formally, $\mathcal{C}$ is collision-free if, for every $C_i, C_j$ where $i \neq j$ and time-step $s \in [1, \vert C_i \vert]$, $\mathcal{W}(C_i)[s] = (v, u)$ and $\mathcal{W}(C_j)[s] = (v', u')$ satisfies $u \neq u'$, $v \neq v'$ and $(v, u) \neq (u', v')$.\looseness=-1

For the remainder of this paper, we assume every robot in the graph is assigned exactly 1 schedule. Given 2 sets of schedules $\mathcal{C}$ and $\mathcal{C}'$, we say $\mathcal{C}$ is \emph{faster} than $\mathcal{C'}$ if $\max_{C_i \in \mathcal{C}} \vert C_i \vert < \max_{C_j' \in \mathcal{C}'} \vert C_j' \vert$. 
Given a graph $G = (V, E)$, set of $k$ robots $R_1, R_2, \dots, R_k$ starting on vertices $sv_1, sv_2, \dots, sv_k$, and set of tasks $\mathcal{T}$, a \emph{fastest} task-completing, collision-free set of $k$-schedules is the set of schedules $\mathcal{C}$ such that any other set of task-completing, collision-free schedules is no faster than $\mathcal{C}$. Note that there may be multiple such schedules.\looseness=-1

\begin{problem}[$k$-\rsp]
    \label{prob:robot-scheduling}
    Given a graph $G = (V, E)$, set of $k$ robots $R_1, R_2, \dots, R_k$ starting on vertices $sv_1, sv_2, \dots, sv_k$, and set of tasks $\mathcal{T}$, what is the fastest task-completing, collision-free set of $k$-schedules $\mathcal{C} = (C_1, C_2, \dots, C_k)$ such that $C_i$ can be assigned to $R_i$, for all $ i \in [1, k]$?\looseness=-1
\end{problem}

 We can rephrase $k$-\rsp as a decision problem by asking, for a given time-limit $L$, does there exist some task-completing, collision-free set of $k$-schedules $\mathcal{C} = (C_1, C_2, \dots, C_k)$ such that $C_i$ can be assigned to $R_i$ and $\vert C_i \vert \leq L$, for all $i \in [1, k]$?\looseness=-1

\begin{example}
An example of a task-fulfilling set of schedules for the graph shown in Figure \ref{fig:example} is
\begin{align*}
   \mathcal{C} = &\{ \left( [(v_7,v_8),(v_8,v_5)],(v_5,5) \right),\\
   &\left([(v_9,v_6),(v_6,v_3),(v_3,v_2)],(v_2,3),[(v_2,v_1),(v_1,v_4)],(v_4,2)\right)\}
\end{align*}
 which has a time span of 10 . However the optimal set of schedules in this case would be: 
 \begin{align*}
\left\{ ([(v_7,v_4)],(v_4,2),[(v_4,v_1),(v_1,v_2)],(v_2,3)),
           ([(v_9,v_6),(v_6,v_5)],(v_5,5)) 
\right\}
\end{align*}
which has a time span of $8$.\looseness=-1
\end{example}


\begin{figure}
    \centering
    \label{fig:example}
\begin{tikzpicture}[scale=0.8]
    \tikzset{vertex/.style = {shape=circle,draw,minimum size=2.5em}}
    \tikzset{edge/.style = {-}}
    \node[vertex,draw=red](v_4) at (0,4){$v_4,{\color{red}2}$};
    \node[vertex](v_1) at (0,6){$v_1$};
    \node[vertex, draw=red](v_5) at (2,4){$v_5,{\color{red}5}$};
    \node[vertex](v_8) at (2,2){$v_8$};
    \node[vertex](v_7) at (0,2) {$v_7$};
    \node [vertex, draw =blue, inner sep=12pt] at (v_7){}; 
    \node[vertex](v_6) at (4,4) {$v_6$};
    \node[vertex](v_9) at (4,2) {$v_9$};
    \node [vertex, draw =blue, inner sep=12pt] at (v_9){}; 
    \node[vertex,draw=red](v_2) at (2,6) {$v_2,{\color{red}3}$};
    \node[vertex](v_3) at (4,6) {$v_3$};
    \draw[edge] (v_4) to (v_1);
    \draw[edge] (v_4) to (v_5);
    \draw[edge] (v_1) to (v_2);R
    \draw[edge] (v_2) to (v_5);
    \draw[edge] (v_2) to (v_3);
    \draw[edge] (v_6) to (v_9);
    \draw[edge] (v_5) to (v_8);
    \draw[edge] (v_8) to (v_9);
    \draw[edge] (v_6) to (v_3);
    \draw[edge] (v_6) to (v_5);
    \draw[edge] (v_8) to (v_7);
    \draw[edge] (v_4) to (v_7);
\end{tikzpicture}
 \caption{A graph with tasks and robots. Blue outlines indicate the positions of robots, red vertices the locations of tasks, and red numbers the durations of the tasks.}
\end{figure}
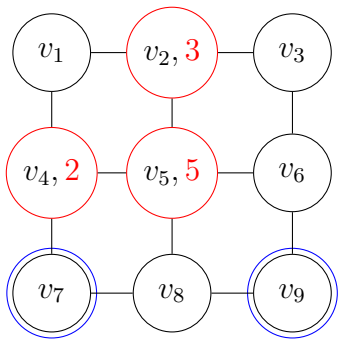


\subsection{Problems used for NP-hardness reductions}

Before providing our results, we provide a quick overview of the problems that are used in Section \ref{sec:np-hardness} as a basis for the hardness. As these are well-known problems, this may primarily be thought of as an overview of the notation used for the remainder of the paper. For more details on this problem, we turn the reader to the textbook of Garey and Johnson \cite{GareyJohnson1979}.\looseness=-1

\paragraph*{$k$-Set Partition}
Given a set of integers $\mathcal{S} = (s_1, s_2, \dots, s_m)$, we define a \emph{partition} of $S$ into $k$ sets as a set of $k$-sets $S_1, S_2, \dots, S_k$ such that $\bigcup_{i \in [1, k]} S_i = \mathcal{S}$ and for any $i, j \in [1, k], i \neq j$, $S_i \cap S_j = \emptyset$. In the case of multiple integers with the same value, we assume that each entry in the set has a unique identifier, allowing this definition to hold. An \emph{exact partition} of $\mathcal{S}$ into $k$ sets $S_1, S_2, \dots, S_k$ is a partition such that $\sum_{s \in S_i} s = \sum_{s' \in S_j} s = \sum_{s'' \in \mathcal{S}} s'' / k$, for every $i, j \in [1, k]$. The $k$-set partition problem asks if an exact partition exists for a given set $\mathcal{S}$ and integer $k$.\looseness=-1

\paragraph*{Hamiltonian Path}
A \emph{Hamiltonian path} for a given graph $G$ is a walk $w$ in $G$ such that each vertex is visited \emph{exactly} once. The Hamiltonian path problem asks if such a path exists for a given graph. For our reduction we consider the more restricted case of finding a Hamiltonian Path starting at a given vertex $v\in G$ \looseness=-1

\section{Hardness Results}
\label{sec:np-hardness}

In this section, we show that the $k$-\rsp problem is NP-complete, even for highly restricted graph classes. Explicitly, we prove NP-hardness results for complete graphs, trees and planar graphs. We note that our hardness result for complete graphs,  and trees hold for at least 2 robots, while that for planar graphs holds even for 1 robot.  As such, the result for trees does not imply the result for planar graphs. The proof of Theorem \ref{thm:general-hardness} follows from these proofs. In order to claim NP-completeness however, we must first prove that the problem is in NP.\looseness=-1

\begin{lemma}
    \label{lem:rsp_in_NP}
    $k$-\rsp is in NP, for any $k \in \mathbb{N}$.\looseness=-1
\end{lemma}

\begin{proof}
    Observe that given any solution to $k$-\rsp, we can verify the correctness in polynomial time 
    (relative to the size of input) by simulating the solution. Hence, the problem is in NP.\looseness=-1

\end{proof}
\begin{theorem}
 \label{thm:clique-hardness}
 $k$-\rsp on complete graphs is NP-complete for any $k \geq 2$.\looseness=-1
 \end{theorem}

 \begin{proof}
 We prove this by a reduction from the set partition problem.\looseness=-1

 Given an instance of the set partition problem containing the set of integers $\mathcal{S} = \{s_1, s_2, \dots, s_m\}$ and integer $k \in \mathbb{N}$, we construct a $k$-\rsp instance containing the complete graph $G = (V, E)$. The vertex set $V$ is composed of two sets $V^R$ and $V^T$ where $V^R = (v_1^r, v_2^r, \dots, v_k^r)$ and $V^T = (v_1^t, v_2^t, \dots v^t_m)$, with $V = V^R \cup V^T$ we refer to $V^R$ as the set of \emph{robot vertices}, and $V^T$ as the set of \emph{task vertices}. As $G$ is a complete graph, $E = V \times V$, i.e. the set of all potential edges corresponding to pairs of vertices $\{(v, u) \mid v, u \in V\}$. We construct the set of tasks $t_1, t_2, \dots, t_m$ where $t_i$ has a time-span of $s_i - 1$ and is located on vertex $v_{i}^t$. Finally, we construct $k$ robots $R_1, R_2, \dots, R_k$ with $R_i$ initially located on vertex $v^r_i$.\looseness=-1

 We claim that there exists a task-completing conflict-free schedule requiring $\sum_{s \in \mathcal{S}} s/k$ time if and only if there exists an exact partition of $S$ into $k$ sets. In one direction, observe that given some exact partition $\mathcal{S}_1, \mathcal{S}_2, \dots, \mathcal{S}_k$ we can construct a schedule $\mathcal{C} = (C_1, C_2, \dots, C_k)$ where, if $\mathcal{S}_i = (S_{i_1}, S_{i_2}, \dots, S_{i_{m_i}})$ then $C_i = ((v^r_i, v^t_{i_1}), t_{i_1}, (v^t_{i_1}, v^t_{i_2}), t_{i_2}, \dots, t_{i_{m_{i} - 1}}, (v^t_{i_m - 1}, v^t_{i_m}), t_{i_m})$. First, observe that, as $\mathcal{S}_1, \mathcal{S}_2, \dots, \mathcal{S}_k$ is an exact partition of $\mathcal{S}$, each vertex in $G$ is visited exactly once, and therefore the set of schedules $\mathcal{C}$ is conflict-free. Further, the time required to complete the schedule $C_i$ corresponds to the time to move between each vertex in the schedule and the time to complete each task. As the task $t_i$ has a duration of $s_i - 1$, then the time for the robot $R$ assigned task $t_i$ in its schedule requires $1$ timestep to reach $v^t_i$ from the previous vertex (either the previous task or the starting vertex), and $s_i - 1$ timesteps to complete $t_i$, the total cost of completing the schedule $C_i$ is equal to $\sum_{s \in \mathcal{S}_i} s$, and thus the total time to complete the set of schedules $\mathcal{C}$ is $\sum_{s \in \mathcal{S}} s / k$.\looseness=-1

 In the other direction, given some set of schedules $\mathcal{C} = (C_1, C_2, \dots, C_k)$ such that $\mathcal{C}$ takes $\sum_{s \in \mathcal{S}} s / k$ timesteps to complete, we can construct a partition of $\mathcal{S}$ by making the sets $\mathcal{S}_1, \mathcal{S}_2, \dots, \mathcal{S}_k$ where $\mathcal{S}_i$ contains the integers in $\mathcal{S}$  corresponding to the tasks completed in $C_i$. Note that as each vertex can be reached in a single timestep from any starting vertex, the total cost of the schedule $C_i$ completing $m_i$ tasks is equal to $m_i$ plus the length of the tasks, equal to $\sum_{s \in \mathcal{S}_i} s - 1$, hence the time-span of $C_i$ is $\sum_{s \in \mathcal{S}_i} s =\sum_{s \in \mathcal{S}} s / k$. Additionally, as $\sum_{i \in [1, k]} \vert C_i \vert \geq \sum_{s \in \mathcal{S}} s$, any schedule taking $\sum_{s \in \mathcal{S}} s / k$ must satisfy $\vert C_i \vert = \sum_{s \in \mathcal{S}} s / k$ for every $i \in [1, k]$, completing the reduction.\looseness=-1
\end{proof}

\begin{theorem}
    \label{thm:star-hardness}
    $k$-\rsp is NP-complete on star graphs for $k \geq 2$.\looseness=-1
\end{theorem}

\begin{proof}
    Recall that a star graph is a tree where all but one vertex has degree 1.\looseness=-1

    We prove this statement by a reduction from the set partition problem.
    Assume we are given a set partition instance where $\mathcal{S} = \{s_1, s_2, \dots, s_{m} \}$. We assume, without loss of generality, that $s \geq 2,$ for all $s \in \mathcal{S}$.
    From this instance, we construct the $2$-\rsp instance as follows.
    Let $G = (V, E)$ be a graph containing the set $V = \{v_s\} \cup V^T \cup V^R = \{v_s\} \cup \{v^t_{1}, v^t_2, \dots, v^t_{m}\} \cup \{v^r_1, v^r_2\}$ of $m + 3$ vertices. We call the subset $V^T =  \{v^t_{1}, v^t_2, \dots, v^t_{m}\}$ the \emph{task vertices}, the subset $V^R = \{v^r_1, v^r_2\}$  the \emph{robot vertices} and vertex $v_s$, the \emph{star vertex}. As the names imply, the task vertices contain the tasks, the robot vertices are the start position of the robots, and the star vertex is the central vertex of the graph. The edge set $E$ is defined as $\{(v_s, v_i) \mid v_i \in V \setminus \{v_s\}\}$. We add 2 robots, $R_1$ and $R_2$, placing $R_1$ on $v^r_1$ and $R_2$ on $v^r_2$.\looseness=-1

    We construct the set of tasks $\mathcal{T} = \{t_1, t_2, \dots, t_{m}\}$, defining the task $t_i$ as having a duration of $2 \cdot s_i - 2$, and is located on $v^t_i$. We highlight now that the even length of each task is key to the remainder of our reduction. In brief, we ensure that there exists some schedule where $R_1$ will complete tasks only on even timesteps and $R_2$ only on odd timesteps (in general for $k$ robots, robot $R_i$ completes tasks on timesteps $t+1$ for $t \equiv i \mod k$ for $i \in [k]$ ) . In this way, we avoid collision as $R_1$ will only occupy the star vertex on odd timesteps and $R_2$ on even timesteps.\looseness=-1

    Now, we claim there exists a schedule $\mathcal{C}$ taking $1 + \sum_{s \in \mathcal{S}} s$ time if and only if there exists a perfect partition of $\mathcal{S}$ into $2$ sets.\looseness=-1


    First, we show that given a schedule $\mathcal{C} = \{C_1, C_2\}$ taking $1 + \sum_{s \in \mathcal{S}} s$ time, we can construct $2$ subsets of $\mathcal{S}, S_1, S_2$ such that $\sum_{s \in S_1} s = \sum_{s' \in S_2} s'$. We do so by adding to $S_i$ the entry in $\mathcal{S}$ corresponding to each task completed in $C_i$. As the task $t_i$ requires $s_i - 2$ time to complete, and the robot assigned to it requires $2$ timesteps to reach $v_i$ from the previous vertex, if the schedule $C$ containing $t_i$ is completed in either $\sum_{s \in \mathcal{S}} s$ or $1 + \sum_{s \in \mathcal{S}} s$ timesteps, then the size of the set $S_1$ must be $\sum_{s \in \mathcal{S}} s / 2$, and hence $S_1$ and $S_2$ are a perfect partition of $\mathcal{S}$. In the other direction, given any perfect partition $S_1, S_2$ of $\mathcal{S}$, we can construct a schedule $\mathcal{C} = \{C_1, C_2\}$ taking $1 + \sum_{s \in \mathcal{S}} s$ timesteps by having $r_1$ complete the tasks $\{t_i \mid \forall s_i \in S_1\}$, and $r_2$ complete the tasks $\{t_j \mid \forall s_j \in S_2\}$, after waiting on the starting vertex for 1 timestep. As the time to travel between each task is $2$, and as each task has an even length, $r_1$ will move only on odd timesteps and $r_2$ on even ones, thus this schedule is collision-free and requires $1 + \sum_{s \in \mathcal{S}} s$ timesteps to complete.

    Now, assume that no schedule taking $1 + \sum_{s \in \mathcal{S}} s$ timesteps exists. Then, by the same arguments as above, it must not be possible to form any perfect partition of $\mathcal{S}$ as such a partition would give a schedule taking $1 + \sum_{s \in \mathcal{S}} s$ time. Hence, this statement holds. Similarly, if no perfect partition of $\mathcal{S}$ exists, then no schedule taking $1 + \sum_{s \in \mathcal{S}} s$ timesteps exists, completing the proof.\looseness=-1

\end{proof}

Despite being NP-hard for $k \geq 2$, when we have only one robot the problem becomes trivial.\looseness=-1
\begin{observation}
    1-\rsp can be solved in polynomial time for star graphs.\looseness=-1
\end{observation}

\begin{corollary}
    $k$-\rsp is NP-complete for trees.\looseness=-1
\end{corollary}

\begin{theorem}
\label{thm:planar-hardness}
1-\rsp on planar graphs is NP-complete, even when all tasks are of equal duration.
\end{theorem}
\begin{proof}
    To prove NP-hardness we reduce from the Hamiltonian Path problem with fixed starting vertex in planar graphs.\looseness=-1

    Let $G = (V, E)$ be a planar graph where $V = \{v_1, v_2, \dots, v_n\}$. From $G$, we construct a 1-\rsp instance with the graph $G$ and set of tasks $\mathcal{T} = \{T_1, T_2, \dots, T_n\}$ where every task has duration 1 and task $T_i$ is placed on vertex $v_i$.\looseness=-1
    

    Now, a robot $A$ is placed on some vertex $v_s$. Observe that the fastest task-completing collision-free schedule for $G$ with $A$ requires visiting every vertex in $G'$ at least once, then spending one timestep on that vertex. Therefore, if a task-completing collision-free schedule $\mathcal{C}$ has $Sp(\mathcal{C}) = 2n - 1$, then there must exist some path visiting every vertex in $V$ exactly once, as visiting any vertex more than once would require an extra timestep. Hence, given such a schedule, there exists a Hamiltonian path starting at $v_s$ in $G$. By checking if any such schedule taking $2n - 1$ timesteps when $A$ starts at vertex $v_s$ for every $v_s \in V$, we can determine if any Hamiltonian path exists in $G$.\looseness=-1

    In the other direction, given a Hamiltonian path starting at $v_s$, we can construct a schedule taking $2n - 1$ timesteps by stopping at each vertex for a single timestep to complete the associated task. Hence, the reduction holds.\looseness=-1
\end{proof}

%

\begin{theorem}
\label{thm:general-hardness}
$k$-\rsp is NP-complete for any $k \in \mathbb{N}$.\looseness=-1
\end{theorem}

%

\section{Algorithmic Results for Path Graphs}
\label{sec:algorithmic-results}

In this section, we present a set of algorithmic results for path graphs. Recall that a graph $G$ is a path if and only if every vertex has a degree at most 2, and there exist exactly 2 vertices with degree 1. Formally, a path $P$ of length $n$ contains the set of vertices $V = \{v_1, v_2, \dots, v_n\}$, and the set of edges $E = \{(v_1, v_2), (v_2, v_3), \dots, (v_{n - 1}, v_n)\}$. 
For the remainder of this section for a given pair of vertices $v_i, v_j$ on a path graph, we say that $v_i$ is \emph{left} of $v_j$ if $i < j$, and that $v_i$ is \emph{right} of $v_j$ if $i > j$.\looseness=-1

In Section \ref{subsec:1-rsp-alg}, we provide an algorithm for finding an optimal schedule for $1$-\rsp on a line. 
 In Section \ref{subsec:2-rsp-alg}, we provide two results regarding 2-\rsp for paths. Explicitly, we provide an algorithm that is optimal when every task has equal duration and a 2-approximation for general length tasks.
In Section \ref{subsec:k-rsp-alg}, we generalise this to give an optimal algorithm for $k$ robots with equal-length tasks and a $k$-approximation in the general setting.\looseness=-1

\subsection{1-\rsp on Path Graphs}
\label{subsec:1-rsp-alg}

In this section, we provide an algorithm for finding the optimal schedule for a single robot on a path. We first provide a sketch of the algorithm, then prove in Lemma \ref{lem:1-rsp-opt} that this algorithm is optimal. Corollary \ref{col:computing-1-rsp-time} shows that the time needed to complete the fastest schedule can be computed via a closed-form expression.\looseness=-1

\paragraph*{1-\rsp Algorithm}

Let $P$ be a path graph of length $n$, let $T = (t_1, t_2, \dots, t_m)$ be a set of tasks, and let $R$ be the single robot starting on vertex $sv = v_{i_s}$. We assume, without loss of generality, that $t_j$ is located on $v_{i_j}$ such that $v_{i_j}$ is left of $v_{i_{j + 1}}$, i.e. $\forall j \in [1, m - 1], i_j < i_{j + 1}$. Note that there may exist some task $t_i$ located on $sv$ without contradiction.
Using this notation, the optimal schedule $\mathcal{C} = \{C\}$ is:\looseness=-1
\begin{itemize}
    \item $C = \{(v_{i_s}, v_{i_s + 1}), \dots$, $(v_{i_m - 1}$, $v_{i_m})$, $t_m$, $(v_{i_m}, v_{i_m - 1})$, $\dots$, $(v_{i_{m + 1} + 1}, v+{i_{m + 1}})$, $t_{m - 1},\dots, (v_{i_1 + 1}, v_{i_1})$, $t_1$ $\}$ if $\vert i_s - i_m \vert \leq \vert i_s - i_1 \vert$.
    \item $C$ $=$ $\{$ $(v_{i_s}, v_{i_s - 1})$, $(v_{i_s - 1}, v_{i_s - 2})$, $\dots$, $(v_{i_1 + 1}, v_{i_1})$, $t_1$, $(v_{i_1}, v_{i_1 + 1})$, $(v_{i_1 + 1}, v_{i_2 + 2})$, $\dots$, $(v_{i_{2} - 1}, v_{i_{2}})$, $t_{2}$, $\dots$, $(v_{i_m - 1}, v_{i_m})$, $t_m$ $\}$ if $\vert i_s - i_m \vert > \vert i_s - i_1 \vert$.\looseness=-1
\end{itemize}

\begin{lemma}
    \label{lem:1-rsp-opt}
    The fastest task-completing schedule for $1$-\rsp on a path graph $P$ of length $n$ with $m$ tasks $T = (t_1, \dots, t_m)$ located on vertices $v_{i_1}, \dots, v_{i_m}$, and a robot $R$ starting on vertex $v_{i_s}$ can be constructed in $O(n)$ time.\looseness=-1
\end{lemma}

\begin{proof}
    We prove this statement by showing that the construction above is correct. Note that if $T$ is not ordered, then we can sort the list by position of the tasks in $O(n)$ using a radix sort. Observe that any task-completing schedule must have the robot completing every task. Therefore, the fastest schedule will correspond to the shortest walk visiting every vertex containing a task. We further assume, without loss of generality, that $i_1 \leq s \leq i_m$, as the fastest schedule for any $1$-\rsp instance where $s < i_1$ (respectively, $s > i_m$) must start with the robot moving from $v_{i_s}$ to $v_{i_1}$, and thus this path can be appended to the final solution.\looseness=-1
    
    Observe that if $v_{i_s}$ is neither $v_{i_1}$ nor $v_{i_m}$, $R$ must visit some subset of vertices more than once. Further, any task-completing schedule must visit both $v_{i_1}$ and $v_{i_m}$ at least once. Therefore, there must exist some subsequence $F$ of the edges in the optimal schedule $C$ corresponding to a walk between $v_{i_1}$ and $v_{i_m}$. Additionally, there must be some subsequence $F'$ corresponding to a walk in the optimal schedule $C$ ending before the first edge in $F$ and corresponding to a walk from $v_s$ to either $v_{i_1}$, or $v_{i_m}$. Therefore, as the above construction only contains these walks, one must be minimal. Now, note that if $\vert i_s - i_m \vert \leq \vert i_s - i_1 \vert$, then the shortest walk from $v_{i_s}$ to $v_{i_m}$ is shorter than the shortest walk from $v_{i_s}$ to $v_{i_1}$, and thus the schedule starting with the walk from $v_{i_s}$ to $v_{i_m}$ is shorter than the schedule starting with the walk from $v_{i_s}$ to $v_{i_1}$. Otherwise, the schedule starting with the walk from $v_{i_s}$ to $v_{i_1}$ is shorter than the schedule starting with the walk from $v_{i_s}$ to $v_{i_m}$.\looseness=-1

\end{proof}

\begin{corollary}
    \label{col:computing-1-rsp-time}
    The  time span of the fastest task-completing schedule for $1$-\rsp on a path graph $P$ of length $n$ with $m$ tasks $T = (t_1, \dots, t_m)$ located on vertices $v_{i_1}, \dots, v_{i_m}$ and a robot $R$ starting on vertex $v_s$ is \[\min(\vert s - i_1 \vert, \vert s - i_m \vert) + i_m - i_1 + \sum\limits_{t \in T} t.\]\looseness=-1
\end{corollary}
 \subsection{2-\rsp on Path Graphs}
 \label{app:2-rsp-alg}

 In this section we discuss $2$-\rsp on a path. First, we provide a new algorithm generalising the above algorithm for $1$-\rsp. In Section \ref{subsec:k-rsp-alg}, we will further generalise this to $k$-\rsp on a path; however, it is valuable to consider $2$-\rsp first, both to illuminate the main algorithmic ideas and to provide a base case for later inductive arguments. As in Section \ref{subsec:1-rsp-alg}, we start by providing an overview of our algorithm, which we call the \emph{partition algorithm}. In Lemma \ref{lem:2-partition-equal-length-tasks}, we show that when all tasks have equal duration, this algorithm is optimal. Finally, in Theorem \ref{thm:2-partition-any-length-tasks}, we show that when there are no bounds on the length of the tasks, this algorithm returns a schedule that has a time-span a factor of at most $2$ greater than the time-span of the fastest task-completing collision-free schedule.\looseness=-1

 \paragraph*{The Partition Algorithm}

 Let $P$ be a path graph of length $n$, let $T = (t_1, t_2, \dots, t_m)$ be the set of tasks, and let $R_{L}$ and $R_{R}$ be the pair of robots starting on vertices $sv_L = v_{i_L}$ and $sv_R = v_{i_R}$ respectively. We call $R_{L}$ the \emph{left robot} and $R_R$ the \emph{right robot}, with the assumption that $sv_L$ is left of $sv_R$. We denote by $i_j$ the index of the vertex containing the task $t_j$, and assume that $i_j < i_{j + 1}$, for every $j \in [1, m - 1]$. We use the notation $P_{i,j}$ for $i,j \in [n], i < j$ to refer to the induced subgraph of $P$ with vertex set $[v_i,v_j]$. For notation, let $C_1(P, T, sv)$ return the optimal schedule for a single robot starting at $sv$ on the path $P$ for completing the task set $T$.\looseness=-1

 We construct the schedule by partitioning the tasks into 2 sets, \newline $T_L = (t_1, t_2, \dots, t_{q})$ and $T_R = (t_{q + 1}, t_{q + 2}, \dots, t_{m})$. We determine the value of $q$ by finding the value which minimises $\max(\vert C_1(P_{1, \max(i_L, i_{q})}, (t_1, t_2, \dots, t_{q}), sv_{L}) \vert,\newline \vert C_1(P_{\min(i_{q + 1}, i_R), m}, (t_{q + 1}, t_{q + 2}, \dots, t_m), sv_R) \vert )$. We will use $C_2(P,T,(sv_L, sv_R))$ to denote the schedule returned by this process.\looseness=-1

 \begin{example}
 An example of execution of the partition algorithm is shown in Figure \ref{fig:2partition_alg_example}. For this example, the left robot (starting on vertex 5) will be assigned the schedule $\left( [(5,4),(4,3)], (3,1), [(3,2),(2,1)],(1,1) \right)$ and the right robot has the schedule $ \left( (6,2),[(6,5),(5,4)], (4,1)  \right )$.\looseness=-1
 \end{example}
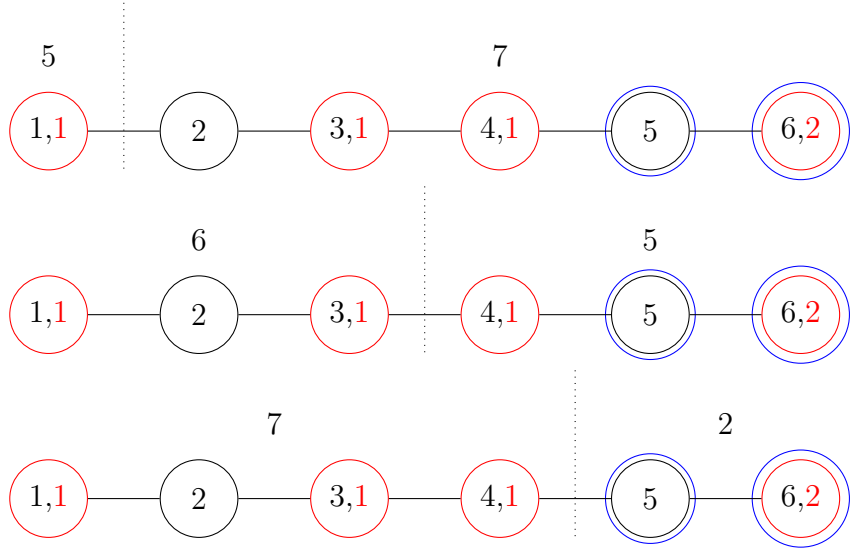
\begin{figure}[h]
\centering
    \begin{tikzpicture}
        \tikzset{vertex/.style = {shape=circle,draw,minimum size=2.5em}}
         
        \tikzset{edge/.style = {-}}
        \node[vertex,draw=red](1) at  (0,0){1,{\color{red}1}};
        \node[vertex](2) at (2,0){2};
        \node[vertex,draw=red](3) at (4,0){3,{\color{red}1}};
        \node[vertex,draw=red](4) at (6,0){4,\color{red}{1}};
        \node[vertex](5) at (8,0){5};
        \node [vertex, draw =blue, inner sep=12pt] at (5){}; 
        \node[vertex,draw=red](6) at (10,0){6,{\color{red}2}};
        \node [vertex, draw =blue, inner sep=13pt] at (6){}; 
        \draw[edge] (1) to (2); 
        \draw[edge] (2) to (3); 
        \draw[edge] (3) to (4); 
        \draw[edge] (4) to (5); 
        \draw[edge] (5) to (6); 
        \draw[dotted] (1,-0.5) to (1,1.75);
        \node at (0,1) {5};
        \node at (6,1) {7};
    \end{tikzpicture}
    \begin{tikzpicture}
        \tikzset{vertex/.style = {shape=circle,draw,minimum size=2.5em}}
         
        \tikzset{edge/.style = {-}}
        \node[vertex,draw=red](1) at  (0,0){1,{\color{red}1}};
        \node[vertex](2) at (2,0){2};
        \node[vertex,draw=red](3) at (4,0){3,{\color{red}1}};
        \node[vertex,draw=red](4) at (6,0){4,\color{red}{1}};
        \node[vertex](5) at (8,0){5};
        \node[vertex, draw=blue, inner sep=12pt] at (5){}; 
        \node[vertex,draw=red](6) at (10,0){6,{\color{red}2}};
        \node[vertex, draw=blue, inner sep=13pt] at (6){}; 
        \draw[edge] (1) to (2); 
        \draw[edge] (2) to (3); 
        \draw[edge] (3) to (4); 
        \draw[edge] (4) to (5); 
        \draw[edge] (5) to (6); 
        \draw[dotted] (5,-0.5) to (5,1.75);
        \node at (2,1) {6};
        \node at (8,1) {5};
    \end{tikzpicture}
    \begin{tikzpicture}
        \tikzset{vertex/.style = {shape=circle,draw,minimum size=2.5em}}
         
        \tikzset{edge/.style = {-}}
        \node[vertex,draw=red](1) at  (0,0){1,{\color{red}1}};
        \node[vertex](2) at (2,0){2};
        \node[vertex,draw=red](3) at (4,0){3,{\color{red}1}};
        \node[vertex,draw=red](4) at (6,0){4,\color{red}{1}};
        \node[vertex](5) at (8,0){5};
        \node [vertex, draw =blue, inner sep=12pt] at (5){}; 
        \node[vertex,draw=red](6) at (10,0){6,{\color{red}2}};
        \node [vertex, draw =blue, inner sep=13pt] at (6){}; 
        \draw[edge] (1) to (2); 
        \draw[edge] (2) to (3); 
        \draw[edge] (3) to (4); 
        \draw[edge] (4) to (5); 
        \draw[edge] (5) to (6); 
        \draw[dotted] (7,-0.5) to (7,1.75);
        \node at (3,1) {7};
        \node at (9,1) {2};
    \end{tikzpicture}
    \caption{An example of the partition algorithm deciding where to split the graph shown on a path $P_6$ with tasks in red and starting vertices of robots being circled in blue. The time span of the schedule $C_1$ is shown above each subgraph.\looseness=-1}
    \label{fig:2partition_alg_example}
\end{figure}
 \begin{lemma}
     \label{lem:contiguous-task-assignment}
     Given an instance of $2$-\rsp on an $n$-length path $P$ with a set of equal-length tasks $T = (t_1, t_2, \dots, t_m)$, and starting vertices $sv_L = v_{i_L}, sv_R = v_{i_R}$, for any schedule $\mathcal{C} = (C_{\ell}, C_r)$ where the rightmost task $t_{L,R}$ assigned to the left robot is right of the leftmost $t_{R,L}$ assigned to the right robot, there exists some schedule $\mathcal{C}' = (C_{\ell}', C_r')$ that takes no more time than $\mathcal{C}$ and does not contain any such tasks.\looseness=-1
 \end{lemma}

 \begin{proof}
     Let $\mathcal{C} = (C_{L}, C_R)$ be a schedule where the rightmost task $t_{L,R}$ (\textbf{L}eft robot's \textbf{R}ightmost task) assigned to the left robot $R_L$ is right of the leftmost task $t_{R, L}$ (\textbf{R}ight robot's \textbf{L}eftmost task) assigned to the right robot $R_R$. Let $t_{L,R}$ be located on $v_{L, R}$ and $t_{R, L}$ be located on $v_{R,L}$. Note that the left robot $R_{L}$ must visit the vertex $v_{R, L}$ containing task $t_{R, L}$, and the right robot $R_{R}$ must visit the vertex $v_{L, R}$ containing $t_{L, R}$. Observe now that $R_{R}$ must be right of $R_{L}$ during the execution of task $t_{L,R}$ by $R_L$. Therefore, if $R_{L}$ completes task $t_{R,L}$ on the last visit to $v_{R,L}$ in the schedule $C_{L}$ before reaching $v_{L, R}$, there can be no conflict with $R_R$. By the same argument, $R_R$ can complete task $t_{L,R}$ on the last visit to $v_{L, R}$ in the schedule $C_R$ before reaching $v_{R,L}$. Hence by assigning $t_{R, L}$ to $R_{L}$ and $t_{L, R}$ to $R_R$ in this manner, there will be no conflict, and further, $R_{L}$ will reach $v_{L, R}$ in the same timestep as $t$ is completed in $C_{L}$, and then immediately leave.\looseness=-1

     Repeating these arguments, we can generate a new schedule $\mathcal{C}'$ taking the same number of timesteps as $\mathcal{C}$ and satisfying the condition that the rightmost task completed by $R_{L}$ is left of the leftmost task completed by $R_{R}$. Note that $R_{L}$ may start right of some task completed by $R_R$ (equivalently, $R_R$ may start left of some task completed by $R_{L}$, though not both) in $\mathcal{C}'$. Thus, it can not be assumed that a faster schedule is formed by removing the walk between the rightmost task assigned to $R_L$ by $\mathcal{C}$, and the rightmost task assigned to it by $\mathcal{C}'$. Hence, we can only claim that the time span $\mathcal{C}'$ is no greater than the time span of $\mathcal{C}$.\looseness=-1
 
\end{proof}

 \begin{lemma}
     \label{lem:2-partition-equal-length-tasks}
     Given an instance of $2$-\rsp on an $n$-length path $P$ with a set of tasks $T = (t_1, t_2, \dots, t_m)$ where the length of $t_i$ is equal to the length of $t_j$ for every $i, j \in [1, m]$. Further, let $sv_{L}$ and $sv_R$ be the starting vertices of the robots. Then $C_2(P, T, (sv_{L}, sv_R))$ is a fastest set of schedules and can be found in $O(m)$ time.\looseness=-1
 \end{lemma}

 \begin{proof}
     Following Lemma \ref{lem:contiguous-task-assignment}, we have that there exists some schedule $\mathcal{C}$ where every task assigned to $R_{L}$ is left of every task assigned to $R_R$ and such that no schedule completes all tasks faster than $\mathcal{C}$. Further, if $R_{L}$ starts right of every task completed in $C_{L}$, then there exists some such $\mathcal{C}$ in which $R_{L}$ starts by moving to the first task completed in $C_{L}$ (equivalently, if $R_R$ starts left of every task completed in $C_{R}$, then there exists some such $\mathcal{C}$ in which $R_R$ starts by moving to the first task completed in $C_R$). Now, note that the fastest schedule solving the given $2$-\rsp instance contains a solution to the $1$-\rsp instances corresponding to $P_{1, \max(\ell, i_q)}, T_{L} = (t_1, t_2, \dots, t_q)$ with the robot starting on $sv_{L}$, and $P_{\min(r, i_{q + 1}, n)}, T_{L} = (t_{q + 1}, t_{q + 2}, \dots, t_m)$ with the robot starting on $sv_{R}$, where $q$ is the number of tasks completed by $R_{L}$. Now, assume that $C_{L}$ is not the fastest schedule satisfying the first $1$-\rsp instance.\looseness=-1
   
     Recalling that the solution given by $C_1(P_{1, \max(sv_L, i_q)}, T_{L}, sv_{L})$ will move $R_{L}$ left if $v_{L}$ is right of the first task in $T_{L}$, \newline $C_{L}$ can be replaced with $C_1(P_{1, \max(sv_L, i_q)}, T_{L}, v_{L})$ without adding any collisions while taking no more time than $C_{L}.$ Following the same arguments for $C_{R}$, we get that the optimal solution to the $2$-\rsp instance must be of the form $$(C_1(P_{1, \max(sv_L, i_q)}, T_{L}, sv_{L}),C_1(P_{1, \min(sv_R, i_{q + 1})}, T_{R}, sv_{R}))$$, for some $q \in [1, m]$. Hence, by checking each value of $q$ and selecting the fastest such schedule, we determine the fastest schedule.\looseness=-1
   
     To achieve the time complexity result, if we assume that the first partition assigns all tasks to the $R_R$, requiring $t_R$ timesteps to complete, then proceeds by removing the leftmost task from the schedule, the time required to complete the second task assigned to $R_R$ can be found in constant time with the equation $t_{R} - \min(\vert r - i_1 \vert, \vert r - i_m \vert ) - i_1 - \vert t_1 \vert + \min(\vert r - i_2 \vert, \vert r - i_m \vert) + i_2$, where $r$ is the index of the vertex $v_r$ where $R_R$ starts. Therefore, after an initial cost of $O(m)$ to compute $T_T$, the time required to complete the schedule assigned to the right robot requires $O(1)$ time at each step. The same arguments may be applied to the time required to compute the schedule assigned to the left robot. As $O(m)$ steps are needed, the time complexity of this method is $O(m)$, and hence the statement holds.\looseness=-1
 
\end{proof}

 Despite being optimal for tasks of equal duration, the partition algorithm does not always return an optimal scheduling. See Figure \ref{fig:alg-counter-example} for an example where this algorithm fails.

   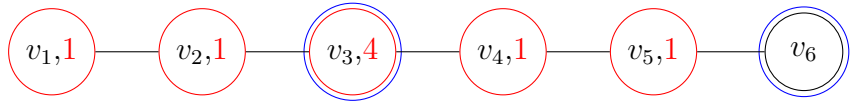
\begin{figure}
    \centering
     \begin{tikzpicture}[scale=1]
        \tikzset{vertex/.style = {shape=circle,draw,minimum size=2.5em}}
        \tikzset{edge/.style = {-}}
        \node[vertex,draw=red](1) at (0,0){$v_1$,{\color{red}1}};
        \node[vertex,draw=red](2) at (2,0){$v_2$,{\color{red}1}};
        \node[vertex,draw=red](3) at (4,0){$v_3$,{\color{red}4}};
        \node[vertex,draw=red](4) at (6,0){$v_4$,\color{red}{1}};
        \node[vertex,draw=red](5) at (8,0){$v_5$,{\color{red}1}};
        \node [vertex, draw =blue, inner sep=13pt] at (3){}; 
        \node[vertex](6) at (10,0){$v_6$};
        \node [vertex, draw =blue, inner sep=12pt] at (6){}; 
        \draw[edge] (1) to (2); 
        \draw[edge] (2) to (3); 
        \draw[edge] (3) to (4); 
        \draw[edge] (4) to (5); 
        \draw[edge] (5) to (6); 
    \end{tikzpicture}
    \caption{Counter-Example to the optimality of the partition algorithm.
    The optimal solution would be the left robot completing the schedule $((v_3,4),[(v_3,v_2),(v_2,v_1)],(v_1,1))$ and the other robot doing $([(v_6,v_5)],(v_5,1),[(v_5,v_4)],(v_4,1),[(v_4,v_3),(v_3,v_2)],(v_2,1))$ with a total time span of 7. Whereas the partition algorithm would return the schedules 
    $((v_3,4),[(v_3,v_2)],(v_2,1),[(v_2,v_1)],(v_1,1))$ and $([(v_6,v_5)],(v_5,1),[(v_5,v_4)],(v_4,1))$ with a total time span of 8.
    }
    \label{fig:alg-counter-example}
\end{figure}

 %

 \begin{theorem}
     \label{thm:2-partition-any-length-tasks}
     Given an instance of $2$-\rsp on an $n$-length path $P = (V, E)$ with a set of tasks $T = (t_1, t_2, \dots, t_m)$ and starting vertices $sv_{L}$ and $sv_R$. Then $C_2(P, T, (sv_{L}, sv_R))$ has time-span at most a factor of 2 greater than the time-span of the fastest set of schedules solving this instance.\looseness=-1
 \end{theorem}

 \begin{proof}
     Let $a$ be the number of timesteps required by the fastest task-completing collision-free schedule to solve this instance, $\mathcal{C} = (C_L, C_R)$. Now, let $t_{R,L}$ be the rightmost task assigned to the left robot. We construct the new schedule $\mathcal{D} = (D_L, D_R)$ by setting $D_L = C_1(P, (t_1, t_2, \dots, t_{R,L}), sv_L)$, and $D_R = C_1(P, T \setminus (t_{L,R +1}, \dots, t_{m}), sv_R)$.
     Note that $R_L$ must visit every vertex containing any task in $(t_1, t_2, \dots, t_{R,L})$ as in $C_L$, $R_L$ must either complete $t_1$ or be left of $t_1$ when it is completed and completes $t_{R,L}$. Therefore, in time $a$, it is possible to complete all tasks in $(t_1, t_2, \dots, t_{R, L})$, and for $R_L$ to visit every vertex in $(t_1, t_2, \dots, t_{R,L})$ at least once. As such, in $2 a$ time, it must be possible for $R_L$ to complete all the tasks in $(t_1, t_2, \dots, t_{R,L})$ and visit each such vertex. As $C_1(P, (t_1, t_2, \dots, t_{R,L}), sv_L)$ returns the fastest schedule in which $R_L$ completes every task in $(t_1, t_2, \dots, t_{R,L})$, $\vert C_1(P, (t_1, t_2, \dots, t_{R,L}), sv_L) \leq 2 a$. Applying the same arguments to the right robot gives $C_1(P, T \setminus (t_1, t_2, \dots, t_{R,L}), sv_R) \leq 2 a$, and thus $C_2(P, T, (sv_L, sv_R))$ has a time-span of at most $2 a$, completing the proof.\looseness=-1
 
\end{proof}
\subsection{$k$-Robot Scheduling on Paths}
\label{subsec:k-rsp-alg}

Now, we generalise this algorithm to the $k$ robot case. To do so, we build a dynamic programming algorithm based on the same principles as the previous partition algorithm. As in the previous sections, we first provide an overview of the algorithm, then the main results. In Theorem \ref{thm:k-robots-partition-equal-length-task}, we show that this algorithm is optimal when all tasks are of equal duration. Finally, in Theorem \ref{thm:k-robot-k-approximation}, we show that this algorithm produces a schedule that takes time at most a factor of $k$ more than the fastest schedule for a given $k$-\rsp instance.\looseness=-1

\paragraph*{The $k$-Partition Algorithm.}

Let $P$ be a path of length $n$, $T = \{t_1, t_2, \dots, t_m\}$ be a set of tasks, and let $sv_1, sv_2, \dots, sv_k$ be the starting vertices of the robots $R_1, R_2, \dots, R_k$ respectively, with the assumption that $R_{i}$ starts left of $R_{i + 1}$, for every $i \in [1, k - 1]$. Further, we denote by $i_{t}$ the index such that $v_{i_t}$ contains task $t$, and assume that $i_{t_j} < i_{t_{j + 1}}$ (i.e. task $t_j$ is left of $t_{j + 1}$) for every $j \in [1, m - 1]$.  We construct a $k \times m$ table $S$, with $S[c, \ell]$ containing the time required to complete the fastest collision-free schedule completing tasks $t_1, t_2, \dots, t_{\ell}$ with robots $R_1, R_2, \dots, R_{c}$. Let $C_1(P, T, sv)$ return the optimal schedule for a single robot starting at $sv$ on the path $P$ for completing the task set $T$, for ease of notation the starting vertex of the robot is often omitted as a parameter. For the purposes of the partition algorithm we also define $C_1(P,\emptyset,sv) := 0$. \looseness=-1

First, observe that $S[1, \ell]$ can be computed, for every $\ell \in [1, m]$, in $O(m)$ time.
Now, assuming the value of $S[c - 1, \ell]$ has been computed for every $\ell \in [1, m]$, the value of $S[c, \ell]$ is computed by finding the value $r$ minimising the larger of $\vert C_1(P, (t_{r + 1}, t_{r + 2}, \dots, t_\ell)) \vert$ and $ S[c-1, r]$, formally

$$S[c, \ell] = \min_{r \in [1, \ell]} \max(\vert C_1(P, (t_{r + 1}, t_{r + 2}, \dots, t_\ell)) \vert , S[c-1, r]).$$
Letting $\mathcal{S}$ be an auxiliary table such that $\mathcal{S}[c, \ell]$ contains the schedule corresponding to the time given in $\mathcal{S}[c, \ell]$. A task-completing collision-free schedule for the $k$-\rsp instance is given in $\mathcal{S}[k, m]$.
For the remainder of this section, let $S_k(P, T, (sv_1, sv_2, \dots, sv_k))$ return the schedule determined by this table. Note that for $S_1(P, T, (sv_1))$, this becomes equivalent to $C_1 (P, T)$
\looseness=-1
\begin{example}
To illustrate this we shall now run the partition algorithm for $k=3$ robots on the instance in Figure \ref{fig:k-partition-example}. In order to arrive at the table $S[c, \ell]$ shown in Table \ref{tab:dynamic-programming} we need to do the following calculations (the intermediate steps of 1 and 2 robots are skipped).\\
$S[3,1] =  \max(S[2,1]=2,0)  =2$ \\
$S[3,2] = \min ( \max(S[2,1]=2,5), \max(S[2,2]=2,0) ) =2$ \\
$S[3,3] = \min ( \max(S[2,1],6), \max(S[2,2],4), \max(S[2,3]=3,0) ) =3$\\
$S[3,4] = \min ( \max(S[2,1],8), \max(S[2,2],6), \max(S[2,3],4),$\\  $\max(S[2,4]=4,0) ) =4$\\
$S[3,5] = \min ( \max(S[2,1],9), \max(S[2,2],7), \max(S[2,3],5),\max(S[2,4],2),$\\ $\max(S[2,5]=6, 0 ) ) =4$\\
$S[3,6] = \min ( \max(S[2,1],10), \max(S[2,2],8), \max(S[2,3],6), \max(S[2,4],3),$\\ $ \max(S[2,5], 1 ), \max(S[2,6]=7,0)  ) =4$\\
\end{example}
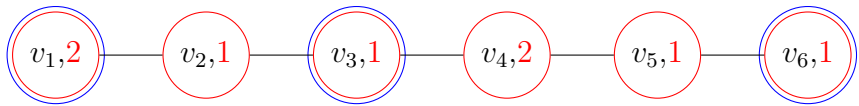
\begin{figure}[h]
\centering
    \begin{tikzpicture}
        \tikzset{vertex/.style = {shape=circle,draw,minimum size=2.5em}}
        \tikzset{edge/.style = {-}}
        \node[vertex,draw=red](1) at (0,0){$v_1$,{\color{red}2}};
        \node[vertex,draw=red](2) at (2,0){$v_2$,{\color{red}1}};
        \node[vertex,draw=red](3) at (4,0){$v_3$,{\color{red}1}};
        \node[vertex,draw=red](4) at (6,0){$v_4$,\color{red}{2}};
        \node[vertex,draw=red](5) at (8,0){$v_5$,{\color{red}1}};
        \node[vertex,draw=red](6) at (10,0){$v_6$,{\color{red}1}};
        \node [vertex, draw =blue, inner sep=13pt] at (1){}; 
        \node [vertex, draw =blue, inner sep=13pt] at (3){}; 
        \node [vertex, draw =blue, inner sep=13pt] at (6){}; 
        \draw[edge] (1) to (2); 
        \draw[edge] (2) to (3); 
        \draw[edge] (3) to (4); 
        \draw[edge] (4) to (5); 
        \draw[edge] (5) to (6); 
    \end{tikzpicture}
    \caption{An instance of the 3-\rsp problem.}
    \label{fig:k-partition-example}
    \end{figure}
\begin{table}[h]
    \centering
   \begin{tabular}{|c|c|c|c|c|c|c|}
             \hline
             {\diagbox{c}{$\ell$}} &1 & 2 & 3 & 4 &5 &6  \\
             \hline
             1 & 2 & 4 & 6 & 9 & 11 & 13\\
             2 & 2 & 2 & 3 & 4 & 6 & 7 \\
             3 & 2 & 2 & 3 & 4 & 4 & 4\\
             \hline
    \end{tabular}
    \caption{The dynamic programming table $s[c,\ell]$ for the instance in Figure \ref{fig:k-partition-example}.}
    \label{tab:dynamic-programming}
\end{table}
    
\begin{figure}[h]
\centering
    \begin{tikzpicture}
        \tikzset{vertex/.style = {shape=circle,draw,minimum size=2.5em}}
        \tikzset{edge/.style = {-}}
        \node[vertex,draw=red](1) at (0,0){$v_1$,{\color{red}2}};
        \node[vertex,draw=red](2) at (2,0){$v_2$,{\color{red}1}};
        \node[vertex,draw=red](3) at (4,0){$v_3$,{\color{red}1}};
        \node[vertex,draw=red](4) at (6,0){$v_4$,\color{red}{2}};
        \node[vertex,draw=red](5) at (8,0){$v_5$,{\color{red}1}};
        \node[vertex,draw=red](6) at (10,0){$v_6$,{\color{red}1}};
        \node [vertex, draw =blue, inner sep=13pt] at (1){}; 
        \node [vertex, draw =blue, inner sep=13pt] at (3){}; 
        \node [vertex, draw =blue, inner sep=13pt] at (6){}; 
        \draw[edge] (1) to (2); 
        \draw[edge] (2) to (3); 
        \draw[edge] (3) to (4); 
        \draw[edge] (4) to (5); 
        \draw[edge] (5) to (6); 
        \draw (3,-0.5) to (3,1);
        \draw (7,-0.5) to (7,1);
    \end{tikzpicture}
    \caption{The partitioning that is output by the partition algorithm for the instance in Figure \ref{fig:k-partition-example}. This would take time 4 as shown in $S[3,6]$ in Table \ref{tab:dynamic-programming}}
    \label{fig:k-partitioning}
\end{figure}
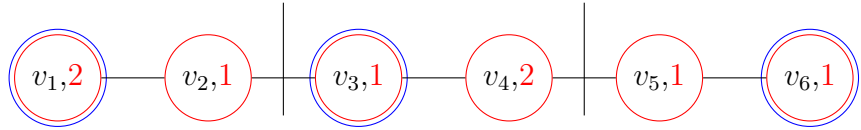

\begin{theorem}
    \label{thm:k-robots-partition-equal-length-task}
    Given an instance of $k$-\rsp on a path $P = (V, E)$ with equal duration tasks $T = (t_1, t_2, \dots, t_m)$ on vertices $v_{i_1}, v_{i_2}, \dots, v_{i_m}$ and $k$ robots $R_1, R_2, \dots, R_k$ starting at $sv_1, sv_2, \dots, sv_k = v_{j_1}, v_{j_2}, \dots, v_{j_k}$, there is no collision-free task-completing schedule for this instance taking less time than the schedule returned by $S_k(P, T, (sv_1, sv_2, \dots, sv_k))$. Further, this schedule can be found in $O(k m n)$ time.\looseness=-1
\end{theorem}

\begin{proof}
    We prove this in an inductive manner, using $S_1(P, T, (sv_1, sv_2))$ as a base case.
    Assume that, for every $c \in [1, k - 1]$, $S_c(P, T, (sv_1, sv_2, \dots, sv_{c}))$ returns such a schedule.
    Now, consider the schedule given by $\mathcal{C} = (C_1, C_2, \dots, C_k) = S_k(P, T, (sv_1, sv_2, \dots, sv_k))$. Let $t_q$ be the leftmost task completed by $R_k$. Note that by construction, the schedule $S_{k - 1}(P, (t_1, t_2, \dots, t_{q - 1}), (sv_1, sv_2, \dots, sv_{k - 1}))$ must be the fastest collision-free schedule completing the tasks $t_1, t_2, \dots, t_{q - 1}$ with the robots $R_1, R_2, \dots, R_{k - 1}$ on $P$.\looseness=-1

    Assume, for the sake of contradiction, that there exists some schedule $\mathcal{C}'$ such that $\mathcal{C}' = (C_1', C_2', \dots, C_k')$ completes all tasks faster than $\mathcal{C}$. If $C_k = C_k'$ then we have a contradiction, as $(C_1', C_2', \dots, C_{k - 1}')$ must then complete $t_1, t_2, \dots, t_{q - 1}$ faster than $(C_1, C_2, \dots, C_{k - 1})$, contradicting the assumption that $S_{k - 1}(P, (t_1, t_2, \dots, t_{q - 1}), (sv_1, sv_2, \dots, sv_{k - 1}))$ is optimal.\looseness=-1

    Now, assume that $R_k$ does not solve $t_m$. Then, either every task is solved by some other robot, or there exists some task $t'$ left of $t_m$ that is solved by $R_k$. 
    Now, if the robot currently assigned to $t_m$ (which we shall call $R_j$)  were to now complete $t_q$ - the leftmost task assigned to $R_k$ (and similarly $R_k$ to complete $t_m$) the new schedule takes at most as long as before, since the tasks are of equal duration and the travel time can only decrease since $R_k$ is the rightmost robot and now completes the rightmost task. 
    Following this argument for each task assigned to $R_k$ from left to right gives a schedule in which $R_k$ solves tasks $t_m, t_{m - 1}, \dots, t_{r}$ for some $r \in [1, m]$. And since the fastest schedule for $R_k$ completing these tasks is given by $C_1(P, ($ $t_{r'}$, $t_{r' + 1}$, $\dots$, $t_m$ $)$, $sv_k)$, and the fastest task-completing schedule for the remaining tasks is given by $S_{k - 1}(P, (t_1, t_2, \dots, t_{r' - 1}), (sv_1, sv_2, \dots, sv_{k - 1}))$, the schedule given by $S_k(P, T, (sv_1, sv_2, \dots, sv_k))$ is therefore optimal.\looseness=-1

    For the time complexity, note that computing the table $\mathcal{S}$ requires $k \cdot m$ entries to be added, each needing $O(m)$ computations corresponding to each partition of the robots (the time span of the fastest schedule for one robot can be calculated in constant time by the formula in Corollary \ref{col:computing-1-rsp-time}) and an additional $O(n)$ time to write the updated schedule. As there are $k \cdot m$ entries, the total time complexity of this process is $O(k m^2 + k m n) = O(k m n)$.\looseness=-1

\end{proof}



\begin{theorem}
    \label{thm:k-robot-k-approximation}
    Given an instance of $k$-\rsp on a path $P = (V, E)$ with tasks $T = (t_1, t_2, \dots, t_m)$ on vertices $v_{i_1}, v_{i_2}, \dots, v_{i_m}$ and robots $R_1, R_2, \dots$, $R_k$ starting at $sv_1, sv_2, \dots, sv_k = v_{j_1}, v_{j_2}, \dots, v_{j_k}$, %
    the schedule returned by \newline $S_k(P, T, (sv_1, sv_2, \dots, sv_k))$ takes time at most a factor of $k$ of the 
    fastest collision-free task-completing schedule for this instance.
\end{theorem}

\begin{proof}
    Let $\mathcal{C} = (C_1, C_2, \dots, C_k)$ be the fastest schedule solving the $k$-\rsp instance, and let $\mathcal{C}' = (C_1', C_2', \dots, C_k')$ be the schedule returned by $S_k(P, T, (sv_1, sv_2, \dots, sv_k))$. Further, let $a$ be the number of timesteps required to complete $\mathcal{C}$. Observe that in $a$ timesteps, there is sufficient time for each robot to complete all tasks assigned to it, as well as relevant movement, including having every robot move between the leftmost and rightmost tasks assigned to it. 
    Therefore, $R_1$ can complete all tasks between the leftmost and rightmost tasks completed in $C_1$ in at most $k \cdot a$ time.
    Repeating this argument gives the $k$ approximation.\looseness=-1

\end{proof}

\subsection{Extension to Cycles and Tadpoles}

We provide a brief extension to $k$-\rsp on cycles. In short, we apply the above algorithm to at most $O(n)$ instances of $k$-\rsp on a path, each formed by removing some distinct edge from the cycle. To prove the correctness of this approach, we provide the following key observation.

\begin{lemma}
    \label{lem:partitioning_the_cycle}
    Given an instance of $k$-\rsp on a cycle $G = (V, E)$ with a set of equal duration tasks $T = \{t_1, \dots, t_m\}$ and robots $r_1, \dots, r_k$, there exists a fastest collision-free task-completing schedule $C$ such that there exists some edge $e \in E$ that is not traversed by any robot in $C$.\looseness=-1
\end{lemma}
\begin{proof}
    We assume the contrary, assume that the vertex $v_{i_j}$ containing $t_j$ satisfies $i_{j-1 \bmod{n} } < i_j < i_{j + 1 \bmod{n}}$, for $j \in [n]$. 
    Similarly, we assume that robots $r_i$ and $r_{i + 1}$ have starting vertices $sv_i$ and $sv_{i + 1}$  respectively, such that $sv_i < sv_{i +1}$, for every $i \in [1, k - 1]$.

    Now, note that for every fastest schedule, for every $i \in [1, k]$, $r_i$ must visit some vertex $v_{x}$ that is also visited by $r_{i + 1 \bmod k}$, and further, $r_i$ must complete some task $t_i$ located on the vertex $v_y$ where $y \geq x$ (or, in the case of $r_k$, the path assigned to $r_k$ uses the edges $$(v_{x}, v_{x + 1 \bmod n}), (v_{x + 1 \bmod n}, v_{x + 2 \bmod n}) \dots (v_{y - 1 \bmod m}, v_y)$$. Now, following the same arguments as above, we can construct a new schedule by assigning to $r_{i}$ the task $t_{i - 1 \bmod k}$, completing the task at the point where $r_i$ visits the vertex containing $t_i$, and removing the task $t_i$, along with the traversal to the vertex containing $t_i$ from the previous vertex in the schedule containing a task completed by $r_i$. Observe that this new schedule takes strictly fewer timesteps than the original schedule as each robot completes the same number of tasks and traverses fewer edges.
\qed
\end{proof}
With Lemma \ref{lem:partitioning_the_cycle}, we can solve the problem of $k$-\rsp on a cycle by checking each of the path graphs formed by removing exactly one edge from the cycle and choosing the best solution.\looseness=-1

\begin{theorem}
    \label{thm:k-robot-cycle}
    Given an instance of $k$-\rsp on a cycle $G = (V, E)$ containing $n$ vertices and the set of tasks $T = \{t_1, t_2, \dots, t_m\}$ all of equal duration, a fastest collision-free task-completing schedule can be found in $O(k m n^2)$ time.\looseness=-1
\end{theorem}
\begin{proof}
    From Lemma \ref{lem:partitioning_the_cycle}, we know that there exists a fastest task-completing schedule for this instance that does not traverse every edge in the graph. Therefore, this schedule is equivalent to a fastest schedule on the path graph $P$ formed by removing some edge $e$ from $G$, where $e$ is not traversed in the schedule. As there are $n$ edges in $G$, we can find a fastest schedule for this instance by finding the fastest task-completing schedule on the graph $P = \{V, E \setminus \{e\}\}$, for any $e \in E$. As finding a fastest schedule for each $P$ takes $O(kmn)$ time, the total complexity of this process is $O(k m n^2)$.
\end{proof}
\begin{corollary}
    For general sets of tasks (i.e. not necessarily having equal duration) the algorithm described in the proof of Theorem \ref{thm:k-robot-cycle} is a $k$-approximation.
\end{corollary}
Finally, we look at \emph{tadpole graphs}. A graph $G = (V, E)$ is a $(m, n)$-tadpole graph if there exists a pair $V_1, V_2 \subseteq V$ such that $V_1 \cap V_2 = \emptyset,\ V_1 \cup V_2 = V$,\ $\vert V_1 \vert = m,$ where the subgraph $(V_1, V_1 \times V_1 \cap E)$ is a cycle, $\vert V_2 \vert = n$, and the subgraph $(V_2, V_2 \times V_2 \cap E)$ is a path, and $\vert V_1 \times V_2 \cap E \vert = 1$. An example of this is given in Figure \ref{fig:tadpole}. We provide two key tools that are used to solve the full problem. First, we show that that we can solve the $2$-\rsp problem on a tree with a single vertex of degree $3$, and every other vertex having degree $2$ or $1$. Secondly, we show that for any instance of $k$-\rsp on a tadpole graph with equal-length tasks, there is an optimal, collision-free schedule where, at most, two robots complete tasks located on both the cycle and the path.\looseness=-1
\begin{figure}
    \centering
    \begin{tikzpicture}[scale=1]
    \tikzset{vertex/.style = {shape=circle,draw,minimum size=2em}}
    \tikzset{edge/.style = {-}}
        \def \n {5}
        \def \radius {1.75cm}
        \def \margin {8} 

        \foreach \s in {1,...,\n}
        {
        \node[vertex] (c_\s) at ({360/\n * (\s - 1)}:\radius) {$c_\s$};
  \draw ({360/\n * (\s - 1)+\margin}:\radius) 
    arc ({360/\n * (\s - 1)+\margin}:{360/\n * (\s)-\margin}:\radius);
}       
     \node[vertex] (p_1) at (2.75,0) {$p_1$};
     \node[vertex] (p_2) at (3.75,0) {$p_2$};
     \node[vertex] (p_3) at (4.75,0) {$p_3$};
     \node[vertex] (p_4) at (5.75,0) {$p_4$};
     \draw[edge] (p_1) to (c_1);
     \draw[edge] (p_1) to (p_2);
     \draw[edge] (p_3) to (p_2);
     \draw[edge] (p_3) to (p_4);
    \end{tikzpicture}
    \caption{(5,4)-Tadpole graph}
    \label{fig:tadpole}
\end{figure}
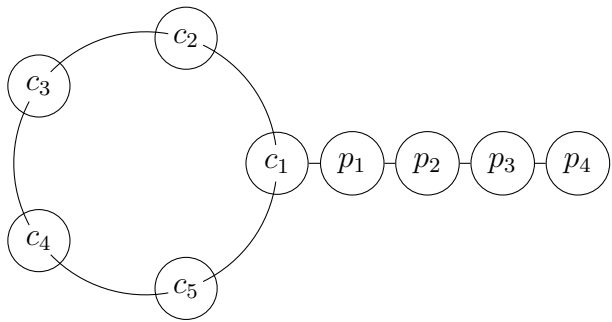
\begin{lemma}
    \label{lem:completing_on_simple_trees}
    Given an instance of $2$-\rsp with $m$ tasks on a tree $T$ where $T$ contains at most one vertex of degree 3, and every other vertex has degree one or two, we can determine a fastest collision-free task-completing schedule in $O(m^2)$ time.\looseness=-1
\end{lemma}
\begin{proof}
    We partition $T$ into two paths $P_1$ and $P_2$ such that $P_1$ is a path containing the starting vertices $v_a$ and $v_b$ of the robots $r_a$ and $r_b$, and the connecting vertex $v_c$, where $d(v_c) = 3$. Observe that, following the same arguments as in Theorem \ref{thm:k-robots-partition-equal-length-task}, there must exists some fastest collision-free task-competing schedule such that there exists a pair of paths $P_{1, a}$ and $P_{1, b}$ where $P_{1,a}$ connects all tasks located on $P_1$ completed by $r_a$, $P_{1, b}$ connects all tasks completed by $r_b$ on $P_1$, and $P_{1, a}$ is not connected to $P_{1, b}$. By the same arguments, such a pair $P_{2, a}$ and $P_{2, b}$ exists for the tasks completed by $r_a$ and $r_b$ respectively on $P_2$. Therefore, the problem becomes determining the optimal partitioning of the tasks between $r_a$ and $r_b$ while maintaining this property, which may be done via brute force in $O(m^2)$ time for $m$ tasks.
\qed
\end{proof}


\begin{lemma}
    \label{obs:only_one_robot_on_both_path_and_cycle}
    Given an instance of $k$-\rsp on a $(m, n)$-tadpole graph with the set of tasks $T$, there exists a fastest collision-free task-completing schedule $C$, where at most two robots $r_a$ and $r_b$ complete a task on both the cycle and the path in $G$. Further, there exists some tree $\mathcal{T}$ spanning all the tasks completed by either $r_a$ or $r_b$, without containing any other node in which a task not completed by either robot is located.
\end{lemma}

\begin{proof}
    Assume, for the sake of contradiction, that no such schedule exists. Let $C$ be some collision-free fastest task-completing schedule where the robots $r_1, \dots, r_i$ visit both the path and the cycle in $G$. First, consider the robots $r_{p_1}, r_{p_2}, \dots, r_{p_j}$ as the robots starting on the path, and further assume that these robots are ordered such that the starting vertex of $r_{p_i}$ is located further from the cycle than that of $r_{p_{i - 1}}$. Then, using the same arguments as in Theorem \ref{thm:k-robots-partition-equal-length-task}, a collision-free task-completing schedule can be constructed in which only a single robot $r_{p_{k}}$ completes a task both on the cycle and the path, and further, for any $k' \in [1, k - 1]$, $r_{p_{k'}}$ only completes tasks on the cycle. Now, consider the robots $r_{c_1}, r_{c_2}, \dots r_{c_{h}}$ such that $r_{c_i}$ starts on the cycle and completes some task both on the cycle and on the path. If there exists at least one robot $r_{p_g}$ that completes a task on the cycle, then using the same arguments as Theorem \ref{thm:k-robots-partition-equal-length-task}, we can trade tasks between the robots until either no robot on the path completes a task on the cycle, or no robot on the cycle completes a task on the path. Now, assume that the set of robots $r_{c_1},r_{c_2}, \dots, r_{c_{h}}$ are the only robots that complete tasks on both the path and the cycle. Then, observe that there exists at most 2 robots, $r_{c_i}$ and $r_{c_j}$ where every robot in $r_{c_1},r_{c_2}, \dots, r_{c_{h}}$ must pass over the starting vertex. Again, we may swap tasks between the robots such that at most two robots complete tasks on both the cycle and the path. Further, these robots must complete the some subset of tasks $T'$ such that there exists a tree $\mathcal{T}$ connecting every vertex containing a task from $T'$, and such that no task in $T \setminus T'$ is located on a vertex in $\mathcal{T}$, completing the proof.

\end{proof}
At a high level, the idea is to partition the set of robots into three subsets, those completing tasks on the cycle, those completing tasks on the path, and those completing tasks on both. Using Lemma \ref{obs:only_one_robot_on_both_path_and_cycle}, we show that we can find a fastest task completing schedule where at most two robots complete a task both on the cycle and on the path. Thus, we end up with at most $O(k^2)$ such sets, noting that any robot between the two robots that are completing tasks on both the cycle and the path must only complete tasks on the path. For each such set, we partition the set of tasks on the cycle (resp. on the path) between the robots completing tasks only on the cycle (resp. on the path), and those completing tasks on both, with $O(m^2)$ possible partitions in the worst case. For each of these partitions, we use Theorems \ref{thm:k-robots-partition-equal-length-task} and \ref{thm:k-robot-cycle} along with Lemma \ref{lem:completing_on_simple_trees} to find an optimal solution to this partition. Finally, we choose the fastest such partition as our solution.\looseness=-1

\begin{theorem}
    \label{thm:k-rsp-on-tadpoles}
    Given an instance of $k$-\rsp on a tadpole graph $G$ with $n$ vertices with $m$ tasks, a fastest collision-free task-completing schedule can be found in $O(k^3 m^4 n)$ time.\looseness=-1
\end{theorem}
\begin{proof}
    We use a similar approach to both path graphs and cycles. From Observation \ref{obs:only_one_robot_on_both_path_and_cycle}, we can partition the set of robots into 3 sets, those only completing tasks on the cycle, those only completing tasks on the line, and those completing tasks on both. We further partition the tasks between these sets. Let $v_c$ be the \emph{connecting vertex}, i.e. the vertex in $G$ with degree 3. 
    
    If the robots $r_a$  and $r_b$ starting on the cycle are the robots completing tasks on both the cycle and path, then any robot starting on some vertex on the path starting at the starting vertex $v_a$ of $r_a$, ending at the starting vertex $v_b$ of $r_b$, and containing the vertex $v_c$, must only complete tasks on the path. Similarly, the tasks completed by $r_a$ and $r_b$ can be defined by a set of three tasks $t_1$ on vertex $v_1$, $t_2$ on vertex $v_2$ and $t_3$ on vertex $v_3$, with all tasks completed by $r_a$ and $r_b$ corresponding to the tasks located on the tree $\mathcal{T}$ connecting $v_1, v_2, v_3$ and $v_c$ such that the path between $v_i$ and $v_j$ passes through $v_c$. Note that a robot starting on some vertex $\mathcal{T}$ must begin the schedule by moving to the path. Therefore, each selection of $r_a, r_b, t_1, t_2$, and $t_3$ partitions the $k$-\rsp instance into three distinct instances. The first instance contains the path from $t_1$ to $t_2$ that does not include $v_c$ (assuming that $t_3$ is the task located on the path), and all robots starting on the path between $v_a$ and $v_b$ that does not contain $v_c$, which may be computed in $O(k m n)$ time using Theorem \ref{thm:k-robots-partition-equal-length-task}.
    
    Second, is the instance on the tree $\mathcal{T}'$, formed by adding to $T$ any edges, and corresponding vertices, on the path between $v_a$ and $v_b$ via $v_c$, which can be solved in $O(m^2)$ time from Lemma \ref{lem:completing_on_simple_trees}.
    
    Finally, we have the instance on the tree $\mathcal{T}'$ formed by the union of the path between $v_a$ and $v_b$ via $v_c$, and the path in $G$. We compute the solution to the last instance by constructing a path $\mathcal{P}$ formed by taking the path $P$ in $G$, and extending with a set of edges $(v_c, v_1), (v_1, v_2), \dots, (v_{\ell - 1}, v_{\ell})$. We place the robots in a greedy manner, with the robot $r$ with the start vertex $v_r$ in the original instance starting at $v_i$, where $i$ the smallest index greater than or equal to the distance between $v_r$ and $v_{c}$ in $\mathcal{T}'$ at which no robot has been placed. In this way, we provide an arbitrary tie breaking mechanism for robots entering the path, such that the number of robots on the path $P$ at any given timestep is equal to the maximum possible for any collision free schedule.
    
    As we can compute the solution to the last instance in $O(k m n)$ time, we therefore can, for a given selection of $r_a, r_b, t_1, t_2, t_3$, construct the fastest collision free schedule in $O(k m n)$ time, noting that $m \leq n$. As there are $m^3 k^2$ possible selections, we require $O(k^3 m^4 n)$ to find a fastest task completing schedule. Note that, if the fastest solution involves a single robot starting on the path completing some task on the cycle, we will find this by the selection of a second robot on the cycle, which will not complete any tasks in the path within the induced sub-instance, completing the proof.

\end{proof}

\section{Conclusion}

We have shown that our definition of $k$-\rsp is hard even on highly constrained classes of graphs while being solvable, with equal duration tasks, for path, cycle, and tadpole graphs as well as $k$-approximable for tasks of any length on path graphs. While these results paint a strong picture of the complexity of this problem, we are left with several open questions. The most direct is as to whether our approximation algorithm for path graphs can be improved or if an optimal algorithm can be found. We conjecture that a polynomial time algorithm exists for this setting; however, at present, no such algorithm has been found. The second natural direction is to look at the remaining classes of graphs that have not been covered by our existing results. The most interesting of these would be lattice graphs, starting with $n \times m$ grids. Such graphs can be used to simulate a wide variety of settings, while still not fitting into any of the classes that are known to be NP-hard. On the other hand, these still provide more complexity than our existing problems due to an exponentially greater number of paths that each robot can take without collision.\looseness=-1

\bibliography{references}
\bibliographystyle{elsarticle-num}
\end{document}